\documentclass[12pt]{article}

\usepackage{epsfig}

\usepackage{amssymb}
\usepackage{amsfonts}

\usepackage{color}
 
\def\be{\begin{equation}}
\def\ee{\end{equation}}
\def\ba{\begin{array}{c}}
\def\ea{\end{array}}

\newcommand{\bea}{\begin{eqnarray}}
\newcommand{\eea}{\end{eqnarray}}

\newcommand{\kt}{\rangle}
\newcommand{\br}{\langle}

\newcommand{\kkt}{\kt\!\kt}

\newtheorem{thm}{Theorem}

\newtheorem{lemma}[thm]{Lemma}

\newenvironment{proof}{\noindent
 {\bf Proof.}}{\hfill$\square$\vspace{3mm}\endtrivlist}

\begin{document}

\begin{center}

{\Large \bf

Factorized Hilbert-space metrics

and

non-commutative
quasi-Hermitian
observables

}

\vspace{10mm}

\textbf{Miloslav Znojil}

\vspace{0.2cm}

\vspace{0.2cm}

The Czech Academy of Sciences, Nuclear Physics Institute,

 Hlavn\'{\i} 130,
250 68 \v{R}e\v{z}, Czech Republic

\vspace{0.2cm}

 and

\vspace{0.2cm}

Department of Physics, Faculty of Science, University of Hradec
Kr\'{a}lov\'{e},

Rokitansk\'{e}ho 62, 50003 Hradec Kr\'{a}lov\'{e},
 Czech Republic

\vspace{0.2cm}

 and

\vspace{0.2cm}

Institute of System Science, Durban
University of Technology, Durban, South
Africa

\vspace{0.2cm}

{e-mail: znojil@ujf.cas.cz}


\end{center}


\section*{Abstract}

In 1992, Scholtz et al (Ann. Phys. 213, p. 74) showed that a set of
non-Hermitian operators can represent observables of a closed
unitary quantum system, provided only that its elements are
quasi-Hermitian (i.e., roughly speaking, Hermitian with respect to
an {\it ad hoc} inner-product metric). We show that such a version
of quantum mechanics admits a simultaneous closed-form
representation of the metric $\Theta_N$ and of the observables
$\Lambda_k$, $k=0,1,\ldots,N+1$ in terms of auxiliary operators
$Z_k$ with $k=0,1,\ldots,N$. At $N=2$ the formalism degenerates to
the well known ${\cal PT}-$symmetric quantum mechanics using
factorized metric $\Theta_2=Z_2Z_1$ where $Z_2={\cal P}$ is parity
and where $Z_1={\cal C}$ is charge.

\subsection*{Keywords}

.

non-standard
descriptions of dynamics
in  Schr\"{o}dinger picture;

apparently non-Hermitian generators of
unitary evolution;

quasi-Hermitian quantum systems with non-commuting
observables;

models with factorized physical Hilbert-space metrics;

\newpage

\section{Introduction}

${\cal PT}-$symmetric quantum mechanics of review \cite{Carl}
offers one of the best known
examples of
technical and conceptual merits of
the use of
a nontrivial Hilbert-space inner-product metric  $\Theta$
which is assumed factorized,
 \be
 \Theta_{\cal PT}={\cal PC} \neq I\,.
 \label{fafa}
 \ee
The
underlying abstract
formulations of quantum
mechanics using any nontrivial
inner-product metric  $\Theta \neq I$
are known as quasi-Hermitian quantum mechanics
\cite{Geyer},
or as pseudo-Hermitian  quantum mechanics
\cite{ali}.
Once one adds the specific factorization ansatz~(\ref{fafa}),
operators ${\cal P}$ and ${\cal C}$
are interpreted,
most often, as parity
and charge, respectively. Still,
also certain less specific forms
of the factors ${\cal P}$ and ${\cal C}$ forming the metric
can be found discussed in the literature \cite{Ali,shendr,Carlbook,PTCCM}.
In our present letter we intend to 
reveal and describe an unexpected new connection between the 
abstract
mathematics of the Hilbert-space-metric
factorization as sampled by Eq.~(\ref{fafa})
and the realistic requirements in 
physics where one often has to demand the generic 
consistent coexistence of several non-commutative
non-Hermitian quantum
observables (say, $\Lambda_j$).

\subsection{The challenge of non-commutativity}

 \noindent
The
challenging nature of all of the
mutually more or less equivalent
approaches
{to quantum dynamics based on
  the nontriviality of the physical Hilbert-space metric (\ref{fafa})
  }
can be illustrated by the
Bagchi's and Fring's
conjecture \cite{framew}
that the use of such a formalism
could enrich even the study of
quantized gravity.
The latter authors argued that
the use of Hilbert spaces
with $\Theta \neq I$
could lead to a new and
consistent description of quantum systems
which admit the existence of
{an observable} minimal length and/or
of some innovative forms
of {experimentally verifiable} uncertainty relations.

{In the latter, most ambitious  
physical project the main technical obstacle 
has later been found to lie in
the necessity of reflecting the
non-commutative nature of the sets of the relevant 
observables \cite{arabky}.
In fact, this was bad news which
still belong among the key motivating forces in
our present letter. The more so because
earlier, 
the obstacle has already been encountered as serious
in the framework of the abstract theory \cite{Geyer}
as well as
in several very concrete calculations
and pragmatic studies, say, in 
condensed matter physics \cite{Dyson} or in
nuclear physics
\cite{Jenssen}.}

{At present, the technical subtlety
of the non-Hermitian non-commutativity
still belongs,
in spite of an enormous recent progress
in the field, among the 
main limiting factors and open questions
in the otherwise highly promising and fairly rapid 
developments in the field
of the non-Hermitian coupled-cluster method 
(see, e.g., the recent progress reports in \cite{PTCCM,Acta}), etc.
}

\subsection{The challenge of technical feasibility}

  \noindent
In general, it is really tempting to expect that
using the quasi-Hermitian (QH) quantum theory
one could really achieve
a conceptual compatibility
between a manifestly non-Hermitian
operator of
momentum (say, $\Lambda_0$) with another,
still non-Hermitian
operator representing position
(say, $\Lambda_1$), etc.
The more explicit
technical realization of such an idea and project
appeared, unfortunately, more difficult than expected.
In \cite{arabky}, in particular,
we demonstrated that
``whenever we are given
more than one candidate for an observable (i.e., say, two operators
$\Lambda_0$ and $\Lambda_1$) in advance'',
the
consistent QH theory
``need not exist in general''.
The methodical as well as phenomenological
scepticism implied by such a discouraging result
is to be weakened in our present letter.

Our present encouragement of a return to optimism
will
rely upon letter \cite{PLA}
in which we proposed a specific
reformulation
of the
${\cal PT}-$symmetric quantum mechanics
of Ref.~\cite{Carl}.
Moreover,
we will also incorporate the results of
paper \cite{EPJP} where we generalized
ansatz (\ref{fafa}) and where we considered
a more flexible version of the QH quantum theory
which admits any number $N$ of
factors forming the metric,
 \be
 \Theta_{N}=Z_{{{N}}}\,Z_{{{N-1}}}\,\ldots\,Z_{2}\,Z_1\,.
  \label{kamenK}
 \ee
After
the transition from the specific ansatz (\ref{fafa})
to its general form (\ref{kamenK}),
and after a brief outline of our basic idea in section \ref{II},
our present main message
will concern the
metric-factorization-related systematic construction of the
sets
 \be
 \Lambda_0\,,\,
 \Lambda_1\,,\,
 \Lambda_2\,,\,
 \ldots\,
 \,
 \label{phph}
 \ee
of the
admissible non-Hermitian (i.e., QH) operators representing the observables.
The details of the construction
will be formulated in sections \ref{III} -- \ref{IVb} and summarized
in section \ref{IV}.

\section{The concept of physical inner-product metric\label{II}}

The origin of the idea of the possible
usefulness and consistency of quantum mechanics
using a nontrivial operator metric $\Theta \neq I$
can be traced back to the year 1992 when the authors
of review paper
\cite{Geyer}
declared
the conventional use of trivial metrics $\Theta_{textbook}=I$
``somewhat restrictive''.
They showed that under suitable mathematical conditions a broader
class of ``consistent quantum mechanical systems'' can be
described using a certain ``set of non-Hermitian operators''
representing observable quantities.
In other words, these operators (see Eq.~({\ref{phph}}) above)
had to be compatible with the
metric-dependent observability {\it alias\,}
quasi-Hermiticity condition
 \be
 \Lambda_k^\dagger\,\Theta = \Theta\,\Lambda_k\,,
 \ \ \ \ \ k=0,1,\ldots
 \,.
 \label{atreat}
 \ee
In this light,
the ${\cal PT}-$symmetric quantum mechanics
can be perceived as
one of the most successful
realizations of the amended formalism
with the formal and trivial observable $\Lambda_0=I$,
with
the observable charge ${\cal C}=\Lambda_1$ and
metric ${\cal PC}=\Lambda_2$ of Eq.~(\ref{fafa}),
and,
for unitary systems, also with
the observable QH Hamiltonian,
$\Lambda_3=H$.

\subsection{Inspiration: Antilinear symmetries}

The inspiration of our present letter
may be traced back to
the Bender's and Boettcher's
conjecture
of an upgrade of
quantum theory \cite{BB} which
made the abstract QH formalism
user-friendlier.
The goal has been achieved
via a rather unconventional
assumption that
a given non-Hermitian
Hamiltonian $H$ can be made acceptable
when required to
exhibit certain auxiliary antilinear symmetries.

These symmetries were of two forms.
One was the
parity-time-reversal
antilinear symmetry {\it alias\,} ${\cal PT}-$symmetry
 \be
 H{\cal PT} = {\cal PT}H\,.
 \label{PTsymm}
 \ee
It combined,
successfully, a phenomenological appeal
of the linear parity $ {\cal P}$ with the antilinear
time-reversal $ {\cal T}$.
In comparison with the abstract QH approach of Scholtz et al \cite{Geyer},
the updated theory based on the {\it ad hoc\,}
technical assumption (\ref{PTsymm})
proved much more
intuitive and more friendly in applications (see, e.g., the collection of
reviews in \cite{Carlbook} for details).


Another, less well known but
much more fundamental antilinear symmetry of the Hamiltonian has
been found in its
parity-charge-time {\it alias\,} ${\cal PCT}$
symmetry \cite{BBJ},
 \be
 H\,{\cal P\,C\,T} = {\cal P\,C\,T}\,H\,.
 \label{pctsy}
 \ee
In a way explained in \cite{ali}
the
action of the time-reversal operator
${\cal T}$ in both of these antilinear symmetries
has precisely the same meaning as
Hermitian conjugation. This means
that Eq.~(\ref{pctsy})
can be given the following equivalent form
 \be
 H^\dagger\,{\cal P\,C} = {\cal P\,C}\,H\,.
 \label{hpctsy}
 \ee
In what follows we will prefer the latter notation convention.

\subsection{Physics: Quasi-Hermitian observables}

Before we proceed to the description of details
let us point out that
the originally purely
formal nature of ansatz (\ref{kamenK})
as introduced in \cite{EPJP}
will be complemented here by the
turn of attention to physics.
The QH quantum dynamics will be assumed based,
in the spirit of Ref.~\cite{Geyer},
not only on the specification of the Hamiltonian $H$
but also
of a multiplet of candidates
for the other, phenomenologically potentially relevant
observables.

Our attention will be paid to
a constructive guarantee that all of the latter operators (\ref{phph})
admit the standard probabilistic experimental interpretation.
According to assumptions as formulated in \cite{Geyer},
our knowledge of $H$ and of the set (\ref{phph}) has to
represent a dynamical input information about the system in question.
All of these operators will have to
satisfy the quasi-Hermiticity
relations
 \be
 \Lambda^\dagger_\ell\,\Theta_N=\Theta_N\,\Lambda_\ell\,,
 \ \ \ \ \ \ell = 0,1, 2, \ldots\,.
 \label{treat}
 \ee
of course.
In paper~\cite{Geyer} we read the reason:
``for the given set of non-Hermitian observables'',
the essence of the QH model-building recipe
lies
in the ``construction of a metric (if it exists)''.
In this context,
the additional, more specific recommendation of the construction
of the metric
as inspired by Eq.~(\ref{fafa}) and
as generalized  in paper \cite{EPJP}
can be read as
a requirement that the metric entering the hidden-Hermiticity
condition (\ref{atreat})
should have
a factorized form.
In our present letter we are going to turn attention
to physics behind such a factorization.

\section{Framework: QH quantum mechanics\label{III}}

In the two most elementary special cases with $N \leq 1$
there is in fact no factorization.
At $N=0$, in particular, the degenerate factorization ansatz (\ref{kamenK})
yields trivial $\Theta_0=I$. This choice
leads to the conventional
quantum mechanics of textbooks. Relation (\ref{treat})
then means that
all of the admissible and eligible observables (\ref{phph})
must be self-adjoint. Their selection (hinted, typically, by the
hypothetical
quantum-classical correspondence \cite{Messiah})
remains formally unconstrained.

\subsection{Bra-ket notation}

At any $N>0\,$ the acceptability of operators (\ref{phph})
is only constrained
by their QH property (\ref{treat}).
One might recall paper \cite{EPJP} and,
using the terminology of functional analysis,
notice that
the traditional $N=0$ Hilbert space of
textbooks
becomes split into the doublet of the two
formally non-equivalent Hilbert spaces ${\cal H}_{phys}$ and ${\cal H}_{math}$
(in \cite{EPJP} we recommended the abbreviations
${\cal H}_{phys}={\cal R}_{0}$ and, at any preselected $N$,
${\cal H}_{math}={\cal R}_{N}$).
The latter space is to be interpreted
as auxiliary, strongly preferred in calculations but
manifestly unphysical. Only
the former Hilbert space ${\cal R}_{0}$
provides the correct probabilistic picture of the quantum system in question.

The less friendly space ${\cal R}_{0}$ differs from
its partner ${\cal R}_{N}={\cal H}_{math}$
by the less elementary inner product,
 \be
 (\psi_a,\psi_b)_{phys}=
 (\psi_a,\Theta\,\psi_b)_{math}\,.
 \label{inprodt}
 \ee
This means that the Dirac's bra-ket notation must be used here with due care.
In what follows, it will still be used
in
the preferred manipulation space ${\cal R}_{N}={\cal H}_{math}$,
 \be
 (\psi_a,\psi_b)_{math}=\br \psi_a|\psi_b\kt\,.
 \ee
The transition to
${\cal R}_{0}={\cal H}_{phys}$ can be then
represented, at any $N$, by formula
 \be
 (\psi_a,\psi_b)_{phys}=\br \psi_a|\Theta_N|\psi_b\kt\,.
 \ee
This means that the measurable predictions
(i.e., the relevant physical matrix elements)
can always be evaluated without ever leaving the user-friendlier
representation space ${\cal R}_{N}={\cal H}_{math}$.

\subsection{The choice of $N=1$ and the QH formalism without factorization}

One returns to
the standard QH formulation of quantum mechanics
when choosing
$N=1$.
In its framework
one deals with a nontrivial Hermitian
and positive definite metric
$\Theta_1$ {\it alias\,} $Z_1$ which is
self-adjoint (i.e.,
$Z_1=Z_1^{\dagger}$)
and, in principle,
eligible as an
observable (so that we can set $\Lambda_1=Z_1$ in (\ref{phph})).
Such an operator becomes nontrivial (i.e., $Z_1 \neq I$) whenever
our preselected observable
quantum Hamiltonian $H$  {\it alias\,} $Z_0$
{\it alias\,} $\Lambda_2$
is chosen non-Hermitian,  $Z_0\neq Z_0^\dagger$. Naturally, as long as the
Hamiltonian carries
information about the closed-system dynamics,
one must know or prove that
its
spectrum
is real, i.e.,
compatible with the unitarity of the evolution.

Once such a spectrum is shown
real (and also discrete, see
the reasons discussed in \cite{scatt}), one arrives at
the fundamental
(one could say  Dieudonn\'{e}'s \cite{Dieudonne})
Hamiltonian-hidden-Hermiticity condition (\ref{treat}),
 \be
  Z_0^\dagger\,Z_1=Z_1\,Z_0\,,\ \ \ \ \ \ \ N=1\,.
 \label{guano}
 \ee
This condition becomes tractable as
an equation to be solved. All of its Hamiltonian-dependent solutions
$\Theta_1=Z_1(Z_0)$
are the formally eligible
metric operators
(see the list of their necessary mathematical properties
as listed in equation Nr. (2.1) in
\cite{Geyer}).
Any one of these Hilbert-space metrics
defines
a different, conceptually consistent quantum system.
All of them would lead to the experimental predictions in a way
which would vary with the different choices of the
sets of the operators (\ref{phph})
representing measurable quantities. In this setting it
makes sense to prove the following simple but important result.

\begin{lemma}
The knowledge of Hilbert-space metric $\Theta=\Theta^\dagger$
enables us to define all of the eligible observables as the operator products
 \be
 \Lambda=M\,\Theta
 \label{insaal}
 \ee
with arbitrary $M=M^\dagger$.
\end{lemma}
\begin{proof}
Referring to the list of the necessary mathematical properties
of the Hilbert-space metrics (see, e.g., \cite{Geyer}),
and having in mind, for the sake of the simplicity of the proof, just the special
models living in finite-dimensional Hilbert spaces,
we leave, to the readers,
all of the necessary care to be paid,
in the general case, to the domains of
the operators, etc
(see, e.g., the dedicated reviews in \cite{book}).
Then, once we insert expression (\ref{insaal})
in the criterion (\ref{treat}) of observability
 $$
 \Lambda^\dagger\,\Theta=\Theta\,\Lambda
 $$
we obtain relation
 $$
 \Theta^\dagger\,M^\dagger\,\Theta=\Theta\,M\,\Theta
 $$
and recall the Hermiticity and invertibility of the metric.
\end{proof}

\section{Quantum models with parity-times-charge metrics\label{IVx}}

At {$N=2$}, the QH-based quantum mechanics is described and discussed in
\cite{PLA}. In factorization formula (\ref{kamenK})
the two new operators
can be interpreted
as traditional
charge ($Z_1={\cal C}$) and parity ($Z_2={\cal P}$),
both, possibly, generalized \cite{shendr}.
In the language of physics the theory
relies upon
a pre-selection of Hamiltonian
$H$
denoted here, alternatively, by the
zero-subscripted symbol $Z_0$
and determining the quantum evolution dynamics
in the QH Schr\"{o}dinger picture,
 \be
 {\rm i}\frac{d}{dt}\,|\psi(t)\kt=Z_0\,|\psi(t)\kt\,,
 \ \ \ \ \ |\psi(t)\kt \in {\cal H}_{math}\,.
 \ee
In \cite{PLA} the corresponding quantum theory with $N=2$ was
dubbed ``intermediate-space Schr\"{o}dinger picture'' (ISP).
From the point of view of mathematics
its three operators $Z_k$ were shown
constrained, solely, by the
triplet of relations
  \be
  Z_0^{\dagger}\,(Z_2\,Z_1)=(Z_2\,Z_1)\,Z_0\,
 \label{able3}
  \ee
 \be
 Z_1^{\dagger}\,Z_2=Z_2\,Z_1
 \label{able2}
 \ee
 \be
 Z_2^{\dagger}=Z_2
 \label{able1}
 \ee
(see Table Nr. 1 in \cite{PLA}).

These equations
guarantee the
unitarity of the evolution.
Directly, this feature of the system is only controlled by
the Hamiltonian-containing Eq.~(\ref{able3})
where one can immediately identify $Z_2\,Z_1\,\equiv\,\Theta_2(H)$.
Still,
the
role of the other two relations is also
related and nontrivial.
For
explanation
it makes sense to return to the ${\cal PT}-$symmetric
quantum mechanics of review \cite{Carl}
in which
a very specific parity-times-charge
factorization of the metric
 \be
 \Theta_2=
 {\cal PC}=Z_2\,Z_1=\Theta^\dagger_2
 \label{hols}
 \ee
has been introduced
and shown useful in applications.
On this background an understanding
of the amendments provided by the
$N=2$ formalism of Ref.~\cite{PLA}
lies in an enhancement of the economy of assumptions. Indeed,
among
relations (\ref{able3}) -  (\ref{able1})
one
does not find
the popular
${\cal PT}-$symmetry assumption (\ref{PTsymm})
even when,
in our present notation, re-written as the
${\cal P}-$pseudo-Hermiticity
condition $H^\dagger \,{\cal P}={\cal P}\,H$.
A minor but interesting generalization of the conventional
${\cal PT}-$symmetric
quantum mechanics is obtained. Even
though the relation between $H$ and ${\cal P}$
remains unspecified,
the ISP theory is consistent.
The operator ${\cal P}$ itself
is purely auxiliary, not carrying any immediate physical
meaning at all.

The latter non-observability
paradox
is caused by the fact that in the framework of ISP we have,
in general,
${\cal P}^\dagger\,\Theta_2\neq \Theta_2\,{\cal P}$.
Such a statement
looks counter-intuitive.
Still, its intuitive acceptability can be supported
by the well known observation that even the coordinate
itself is not an observable quantity in general \cite{arabky}.
Although the coordinate is often
assumed to be measurable,
the construction of its operator representation
appeared to be a difficult task
even in the most elementary QH square-well models
(see, e.g., an example of the construction in \cite{Batal}).
As a consequence, the generalized parity operator
constrained by
Hermiticity (\ref{able1})
is just a freely variable ``parameter''
specifying the quantum model
via its observables \cite{Geyer}.

Once we accept
the ISP Hamiltonian $Z_0$
as an observable which is given in advance, we
are prepared to specify
the necessary Hermitian metric (\ref{hols}) as
one of the eligible solutions of Eq.~(\ref{able3}).
This step being completed,
we are left with the single mathematical constraint (\ref{able2}).
Nevertheless, it is trivial to see that this is just
an identity, i.e., just an
equation which is equivalent to the
Hermiticity of the solution $\Theta_2(H)$ of Eq.~(\ref{able3}).
This means that the acceptable generalized charge $Z_1={\cal C}$
is not arbitrary, being
fully defined in terms of the metric and of the preselected parity,
 \be
 {\cal C}={\cal P}^{-1}\,\Theta_2\,,
 \ \ \ \ \ N=2\,.
 \label{chacha}
 \ee
The multiplication of the identity~(\ref{able2}) by $Z_1$
from the right yields the new, equivalent formula
 \be
 {\cal C}^{\dagger}\,\Theta_2=\Theta_2\,{\cal C}\,,
 \ \ \ \ \ N=2\,.
 \label{uhols}
 \ee
It says that operator (\ref{chacha})
called generalized charge is
not only ${\cal PT}-$symmetric (see Eq.~(\ref{able2}))
but also
quasi-Hermitian. It
{\em can\,} consistently be used as representing
a measurable physical quantity. In Eq.~(\ref{phph})
we may identify, in such a case,
$\Lambda_1=Z_1={\cal C}$.
Formula (\ref{chacha}) can be read as a sample of
Eq.~(\ref{insaal}) with $M\,\equiv\,{\cal P}^{-1}$.
Finally, an entirely analogous argument implies that
also the metric itself could play an analogous role of
another observable, $\Lambda_2=Z_2\,Z_1=\Theta_2$.
Once more, it is also possible to add $\Lambda_3=Z_0=H$.

We may conclude that at $N=2$,
a fully consistent unitary quantum system is obtained.
In terms of the input information represented by the ``suitable physical''
operator $H$ and its
``arbitrary mathematical'' operator-valued-parameter complement ${\cal P}$,
we are able to define the metric $\Lambda_2=\Theta_2(H)$
(as any solution of Eq.~(\ref{able3}))
and,
via Eq.~(\ref{chacha}), the charge,
$\Lambda_1={\cal C}$.
Such an explicit constructive specification
of the
operators of observables converts all of the
three obligatory consistency
constraints (\ref{able3}), (\ref{able2}) and (\ref{able1})
into mathematical identities.



\section{Nonstandard quantum models with $N =3$\label{IVa}}



Table Nr. 2 in \cite{EPJP} samples the one-to-one correspondence
between the factorization of metric
and the
unitarity of evolution of a quantum system in the
generalized Schr\"{o}dinger picture (GSP) using a preselected $N$.
At $N=3$, with the physical Hilbert-space metric $\Theta_3=Z_3\,Z_2\,Z_1$,
the internal consistency of the
GSP quantum theory has been shown guaranteed by the
quadruplet of constraints
 \be
 Z_3^{\dagger}=Z_3
 \label{bable1}
 \ee
 \be
 Z_2^{\dagger}\,Z_3=Z_3\,Z_2
 \label{bable2}
 \ee
  \be
  Z_1^{\dagger}\,(Z_3\,Z_2)=(Z_3\,Z_2)\,Z_1\,
 \label{bable3}
  \ee
  \be
  Z_0^{\dagger}\,\Theta_3=\Theta_3\,Z_0\,.
 \label{bable4}
  \ee
The physical meaning of the last item (\ref{bable4}) is obvious
because it guarantees the hidden Hermiticity of
our preselected Hamiltonian $H=Z_0$. Using
this equation one can determine and select one of the eligible
solutions $\Theta_3=\Theta_3(Z_0)$ of this equation
and assign it the role of the
correct physical Hilbert-space metric.

The first, simplest constraint (\ref{bable1}) is
imposed upon the
``generalized parity'' $Z_3={\cal P}$ which remains,
in a parallel to the previous $N=2$ scenario, unobservable.
This enables us to treat such a (necessarily, self-adjoint)
operator, as before,
as a carrier of a ``dynamical input information''.
Its unconstrained variability
can be used to characterize
the differences between
various phenomenologically non-equivalent quantum systems which are
sharing the same Hamiltonian (i.e., the same
form of the operator of energy).

The interpretation of the next constraint (\ref{bable2})
remained purely formal in \cite{EPJP}.
The product $Z_3Z_2$ has been identified there with a new operator
$Y_3$ which happens to be self-adjoint,
$Y_3=Y_3^\dagger$. This operator cannot be treated as an
observable because
$Y_3^\dagger\,\Theta_3 \neq \Theta_3\,Y_3$ in general.
It enters the GSP theory, therefore, simply as
another, freely variable operator parameter
of the model.
From this point of view, Eq.~(\ref{bable2})
degenerates to an identity in which,
in terms of the already available operators $Z_3$ and $Y_3$,
the unique ``unobservable quasiparity'' $Z_2={\cal Q}$
is defined as follows,
 $$
 {\cal Q}=Z_3^{-1}\,Y_3\,.
 $$
In the last step we are left with constraint (\ref{bable3}).
After its pre-multiplication by
(presumably, non-singular factor) $Z_1$ from the right it acquires the
equivalent form
 $$
 Z_1^\dagger\,\Theta_3=\Theta_3\,Z_1\,.
 $$
This confirms that $Z_1$ can represent an observable.
For the sake of definiteness
let us speak about a renormalized charge, $Z_1={\cal R}$.
As long as it is uniquely defined,
 \be
 {\cal R}=Y_3^{-1}\,\Theta_3\,
 \ee
we may recall the
list of the potentially eligible observables~(\ref{phph})
and identify $\Lambda_1={\cal R}$.

In such an upgraded notation our $N=3$ Eq.~(\ref{bable4})
yields the
solution $\Theta_3={\cal PQR}= \Theta_3(H)$
which is required self-adjoint, but this
constraint appears precisely equivalent to
our last relation~(\ref{bable3}).
In other words, the whole set of constraints (\ref{bable4}) - (\ref{bable3})
becomes
converted into identities.
In their light it is then easy to
identify the other two eligible observables, viz.,
$\Lambda_2={\cal QR}$ and
$\Lambda_3={\cal PQR}$ and, finally, $\Lambda_4=H$ - see \cite{EPJP}.

In the latter paper, unfortunately,
the scenarios with $N > 3$
remained unexplored. Here, we will
describe the
explicit constructive
identification of the potential operators of
the phenomenologically relevant observables at any preselected
number $N$ of factors in the general factorized metric of Eq.~(\ref{kamenK}).


\section{Nonstandard quantum models with arbitrary $N$\label{IVb}}

In paper \cite{EPJP}
the quasi-Hermiticity property
  \be
  H^{\dagger}\,\Theta_N=\Theta_N\,H\,
 \label{deblesep}
  \ee
of a preselected Hamiltonian
$H=Z_0$ was assumed satisfied by the factorized
Hilbert-space metric
$\Theta_N=Z_N\,Z_{N-1}\,\ldots\,Z_2\,Z_1\,$ at an arbitrary $N$.
The separate factors $Z_j$ were required to
obey the set of the GSP theoretical consistency requirements
   \be
  Z_1^{\dagger}\,(Z_N\,\ldots\,Z_3\,Z_2)
  =(Z_N\,\ldots\,Z_3\,Z_2)\,Z_1\ \ (=A_N\,\equiv\,\Theta_N),
 \label{deble4}
  \ee
   \be
  Z_2^{\dagger}\,(Z_N\,\ldots\,Z_4\,Z_3)=(Z_N\,\ldots\,Z_4\,Z_3)\,Z_2\ \ (=B_N),
 \label{deble4b}
  \ee
$$\ldots$$
  \be
  Z_{N-2}^{\dagger}\,(Z_N\,Z_{N-1})=
  (Z_N\,Z_{N-1})\,Z_{N-2}\ \ (=X_N),
 \label{deble3}
  \ee
  \be
  Z_{N-1}^{\dagger}\,Z_N=Z_N\,Z_{N-1}\ \ (=Y_N),
 \label{deble3b}
  \ee
 \be
 Z_N^{\dagger}=Z_N\,.
 \label{deble1}
 \ee
Recalling the last item (\ref{deble1}) we deduce that $Y_N=Y_N^\dagger$
and, subsequently, that $X_N=X_N^\dagger$ and so on, until
$B_N=B_N^\dagger$ and $A_N=A_N^\dagger$. One of the
most important subtle consequences of this observation is
the following result.

\begin{thm}
\label{theoman}
The set (\ref{treat})
of conditions of the quantum-theoretical observability
is satisfied by
the set of the operator products
  $
  \Lambda_k=Z_k\,Z_{k-1}\,\ldots\,Z_2\,Z_1\,
  $
with $ k=1,2,\ldots,N\,$.
\end{thm}
  \begin{proof}
First we notice that the right-hand-side of Eq.~(\ref{deble4})
is equal to the metric, $A_N=\Theta_N$. What is less obvious is that
after the
multiplication of the next relation (\ref{deble4b}) by
$Z^\dagger_1\,\equiv\,\Lambda_1^\dagger$ from the left
we may recall the previous identity and obtain the  metric as well,
$\Lambda_1^\dagger\,B_N=\Theta_N$. Similarly, we have
$\Lambda_2^\dagger\,C_N=\Theta_N$ etc., until
$\Lambda_{N-1}^\dagger\,Z_N=\Theta_N$.
Every such a relation
has the form
 $$
 \Lambda_k^\dagger\,M_k=\Theta_N\,,
 \ \ \ \ \ k=0,1,\ldots,N-1
 $$
where $M_0=A_N$, $M_1=B_N$, etc.,
and
where we use $\Lambda_0=I$ at $k=0$.
After we multiply
each of these equations by $\Lambda_{k+1}$ from the right,
we reveal that $\Lambda_0^\dagger\,M_0\,\Lambda_{1}=\Lambda_{1}^\dagger
\,\Theta_N$, $\Lambda_1^\dagger\,M_1\,\Lambda_{2}=\Lambda_{2}^\dagger
\,\Theta_N$ and, in general,
$\Lambda_k^\dagger\,M_k\,\Lambda_{k+1}=\Lambda_{k+1}^\dagger
\,\Theta_N$. In other words,
we reveal that the resulting sequence of equations coincides with
the QH conditions~(\ref{treat}).
  \end{proof}
From this result we may immediately deduce that
in the GSP formulation of quantum mechanics at a fixed $N$,
the system under consideration is characterized, first of all, by its
Hamiltonian $H$ (which is assumed observable) and by the Hilbert-space metric
which is, due to constraint~(\ref{deblesep})
(i.e., due to Eq.~(\ref{atreat}) at $k=0$), Hamiltonian-dependent,
$\Theta=\Theta_N(H)$.

In the light of theorem \ref{theoman}
the subscript $N$
characterizes not only the
number
of factors in the metric-operator ansatz (\ref{kamenK})
but also
the size of the
multiplet of the
invertible and freely variable self-adjoint
operator-parameters $Z_N, Y_N, X_N, \ldots, B_N$
of the model
as well the size of the multiplet of the
observables defined in closed form. This is our final result:
their
first few $N-$dependent lists are sampled here
in the form of the first few columns of Table~\ref{o3x}.

\begin{table}[h]
\caption{Observables $\Lambda_k$
available
for a given QH Hamiltonian $H$, integer $N$, Hilbert-space metric $\Theta_N(H)$
and for an $(N-1)-$plet of  invertible self-adjoint operator-parameters $Z_N, Y_N, X_N, \ldots,B_N$.
The first-line items $\Lambda_0$ are trivial since
$Y_2\, \equiv\,\Theta_2(H)$ at $N=2$, etc. At any $N$, we also recalled
the observability status of the ``input'' Hamiltonian and added the symbol
$\Lambda_{N+1}\, \equiv\,H$.
}
 \label{o3x} \vspace{.4cm}
\centering
\begin{tabular}{||c|ccccc||}
    \hline \hline
    $N$ & 2&3&4&5&\ldots
 \\
 \hline \hline
    &&&&&\\
  $\Lambda_0$& $I$& $I$&
   $I$& $I$&
  \ldots
   \\
  $\Lambda_1$& $Z_2^{-1}\Theta_2(H)$& $Y_3^{-1}\Theta_3(H)$& $X_4^{-1}\Theta_4(H)$&
   $W_5^{-1}\Theta_5(H)$&
  \ldots
   \\
  $\Lambda_2$& $\ \ \ \Theta_2(H)$& $Z_3^{-1}\Theta_3(H)$& $Y_4^{-1}\Theta_4(H)$& $X_5^{-1}\Theta_5(H)$&
  \ldots
   \\
  $\Lambda_3$&$H$& $\ \ \ \Theta_3(H)$& $Z_4^{-1}\Theta_4(H)$& $Y_5^{-1}\Theta_5(H)$&\ldots
   \\
  $\Lambda_4$&-&$H$& $\ \ \ \Theta_4(H)$& $Z_5^{-1}\Theta_5(H)$& \ldots
   \\
  $\Lambda_5$&-&-&$H$& $\ \ \ \Theta_5(H)$& \ldots
   \\
   $\vdots$ &&& &$ \ddots\ \ \ $ $ \ \ \ $&$\ddots$
    \\
\hline \hline
\end{tabular}
\end{table}



\section{Discussion\label{IV}}

A decisive technical merit of all of the existing
QH quantum models is that
their correct
physical Hilbert space ${\cal H}_{phys}$
is so easily represented in another,
extremely user-friendly
Hilbert space
${\cal H}_{math}$.
After the mere amendment
$\br \psi_a|\psi_b\kt\to \br \psi_a|\Theta|\psi_b\kt$ of the inner product
the experimental predictions of the
theory acquire the standard probabilistic form
which is, naturally, strictly equivalent to
its conventional (but, by assumption, prohibitively complicated)
textbook Schr\"{o}dinger-picture
alternative \cite{Geyer,ali}.

Technically, this means that
the unitary evolution
of a quantum system in question
has to be described
not only by the kets $|\psi(t)\kt$ and by
the corresponding
conventional time-dependent
Schr\"{o}dinger equation
 \be
 {\rm i}\frac{d}{dt}\,|\psi(t)\kt=H\,|\psi(t)\kt\,,
 \ \ \ \ \ |\psi(t)\kt \in {\cal H}_{math}
 \ee
but also, in parallel,
by the metric-multiplied kets
$|\psi(t)\kkt=\Theta_N\,|\psi(t)\kt$
with the evolution controlled
by another Schr\"{o}dinger equation in which
the Hamiltonian is replaced by
its Hermitian conjugate,
 \be
 {\rm i}\frac{d}{dt}\,|\psi(t)\kkt=H^\dagger\,|\psi(t)\kkt\,,
 \ \ \ \ \ |\psi(t)\kkt \in {\cal H}_{math}\,.
 \ee
In review \cite{ali}
it has been emphasized that
the use of the QH language
helped to elucidate several deep and
long-standing unresolved conceptual as well as technical puzzles, say,
in the field of relativistic
quantum mechanics or in quantum cosmology.
In several other physical contexts, unfortunately,
the abstract version of the QH theory
has been found
''very difficult to implement'',
with reasons explained on p. 1216
of Ref.~\cite{ali}.
In this context, our present letter is to be read as
the description of one of the
innovative simplification strategies.

In a final remark let us add that
the main starting-point technicality of the GSP approach, viz. the
choice of the positive-definite
solution $\Theta_N=\Theta_N(H)$
of the constraint (\ref{deblesep})
is, in general,
ambiguous \cite{SIGMAdva}.
In this sense,
any choice can be considered and used
as a starting point of the GSP-based metric-multiplication
strategy as summarized in Table \ref{o3x}.
Naturally, every such an initial selection of the
physics-determining
operator $\Theta_N(H)$
must
satisfy all of the obligatory mathematical properties
as listed
and discussed, say, in review \cite{Geyer}.
In \cite{PLA,EPJP} we emphasized that
among them a key role is played by the mathematical requirements
imposed upon the separate factors $Z_k$ of the metric.
In the present paper the emphasis has been shifted to the
physical aspects of these factors. We revealed that
the phenomenological
information carried by these factors
is given by Theorem \ref{theoman},
i.e., by their re-interpretation as factors in the
candidates for observables
 \be
  \Lambda_k=Z_k\,Z_{k-1}\,\ldots\,Z_2\,Z_1\,
  \ee
with $ k=1,2,\ldots,N\,$.

%
%
%
%
%
%
%
%
%
%
%
%


\newpage

\begin{thebibliography}{10}

\bibitem{Carl}
C. M. Bender, Making Sense of Non-Hermitian Hamiltonians.
Rep. Prog. Phys. 70, 947 (2007).

\bibitem{Geyer}
F. G. Scholtz, H. B. Geyer and F. J. W. Hahne,
Quasi-Hermitian Operators in Quantum
Mechanics and the Variational Principle.
Ann. Phys. (NY) 213,
74
(1992).



\bibitem{ali}
A. Mostafazadeh,
 Pseudo-Hermitian representation of Quantum Mechanics.
Int. J. Geom. Meth. Mod. Phys. 7,
1191
(2010).




\bibitem{Ali}
 A. Mostafazadeh, J. Math. Phys. 43, 205 (2002).


\bibitem{shendr}
M. Znojil and H. B. Geyer, Smeared quantum lattices
exhibiting PT-symmetry with positive P. Fortschr. Phys. - Prog.
Phys. 61, 111 (2013).

\bibitem{Carlbook}
C. M. Bender, PT Symmetry in Quantum and Classical Physics.
(World Scientific, Singapore, 2018).



\bibitem{PTCCM}
R. F. Bishop and M. Znojil, Non-Hermitian coupled cluster
method for non-stationary systems and its interaction-picture reinterpretation.
Eur. Phys. J. Plus 135, 374 (2020).



\bibitem{framew}
B. Bagchi and A. Fring,
Minimal length in Quantum Mechanics and non-Hermitian Hamiltonian systems.
Phys. Lett. A 373, 4307
(2009).
%

\bibitem{arabky}
M. Znojil, I. Semor\'adov\'a, F.  R\r{u}\v{z}i\v{c}ka, H.
Moulla and I. Leghrib, Problem of the coexistence of several non-Hermitian
observables in PT-symmetric quantum mechanics. Phys. Rev. A 95, 042122
(2017).

\bibitem{Dyson}
  F. J. Dyson, Thermodynamic Behavior of an Ideal Ferromagnet. Phys.
Rev. 102, 1230 (1956).

\bibitem{Jenssen}
  D. Janssen, F. D\"{o}nau, S. Frauendorf and R. V. Jolos,  Boson
description of collective states.
Nucl. Phys.
A 172, 145
 - 165
(1971).

\bibitem{Acta}
R. F. Bishop and M. Znojil, The Coupled-Cluster Approach to Quantum Many-Body Problem
                   in a Three-Hilbert-Space Reinterpretation.
 Acta Polytech. 54, 85-92
  (2014).

\bibitem{PLA}
M. Znojil,
 Quantum mechanics using two auxiliary inner products.
 Phys. Lett. A 421, 127792  (2022).


\bibitem{EPJP}
M. Znojil,
Feasibility and method of multi-step Hermitization of crypto-Hermitian quantum Hamiltonians.
Eur. Phys. J. Plus 137, 335 (2022).

\bibitem{BB}
C. M. Bender and S. Boettcher, Real spectra in non-Hermitian Hamiltonians having PT symmetry.
Phys. Rev. Lett. 80,  5243
(1998).


\bibitem{BBJ}
C. M. Bender, D. C. Brody and H. F. Jones, Phys. Rev. Lett. 89, 270401  (2002)
and Phys. Rev. Lett. 92,
 119902 (2004) (erratum).

\bibitem{Messiah}
A. Messiah, Quantum Mechanics I  (North Holland, Amsterdam, 1961).

\bibitem{scatt}
M. Znojil, Scattering theory with localized non-Hermiticities.
Phys. Rev. D 78, 025026 (2008).

\bibitem{Dieudonne}
J. Dieudonne,
Proc. Int. Symp. Lin. Spaces
(Pergamon, Oxford, 1961), pp. 115 - 122.





\bibitem{book}
F. Bagarello, J.-P. Gazeau, F. Szafraniec and M. Znojil, Eds.,
Non-Selfadjoint Operators in Quantum Physics: Mathematical Aspects
(Wiley, Hoboken, 2015).


\bibitem{Batal}
A. Mostafazadeh and A. Batal,
Physical Aspects of Pseudo-Hermitian and PT-Symmetric Quantum Mechanics.
%
%
%
J. Phys. A: Math. Gen. 37, 11645
(2004).




\bibitem{SIGMAdva}
M. Znojil, On the role of the normalization factors
$\kappa_n$ and of the pseudo-metric P in crypto-Hermitian
quantum models. {Symm. Integ. Geom. Methods and Appl.}  4, 001 (2008)
(e-print overlay:
arxiv:0710.4432v3).


\end{thebibliography}
\end{document}